\documentclass[letterpaper,12pt]{article}
\usepackage{amssymb, amsmath}
\usepackage{amsthm}
\usepackage{algorithm}
\usepackage{algorithmic}
\usepackage{enumerate}
\usepackage[numbers,square,comma,sort&compress]{natbib}
\usepackage[pdftex,colorlinks]{hyperref}
\newtheorem{theorem}{Theorem}

\newtheorem{corollary}[theorem]{Corollary}

\newtheorem{lemma}[theorem]{Lemma}
\newcommand{\comment}[1]{}
 
\newcommand\bs[1]{\boldsymbol{#1}}
\begin{document}
%\title{Efficient approximations of closeness centrality}
% \footnote{The authors
%are supported in part by the National Science Council of Taiwan
%under grant 97-2221-E-002-096-MY3 and Excellent Research Projects of
%National Taiwan University under grant
%%97R0062-05.
%98R0062-05.
%}}
%\title{Sublinear bounds for $1$-median selection in metric spaces}
%\title{A Las Vegas approximation algorithm for metric $1$-median selection}
\title{Deterministic metric $1$-median selection with very few queries
%\footnote{Supported
%in part by the Ministry of Science and Technology of Taiwan under
%grant 110-2221-E-155-012-.}
\footnote{Part of this paper appears in {\em Proceedings of the 27th International
Computing and Combinatorics
Conference} (COCOON~2021).}}
%\title{Metric $1$-median selection via a new pseudo-metric}
%\title{Deterministic sublinear-time algorithms for $1$-median selection}
%\text{http://par.cse.nsysu.edu.tw/$\sim$algo/paper/paper13/10.pdf}}}
%\title{On the approximability of metric $1$-median selection by nonevasive algorithms}
%\title{Metric $1$-median selection revisited}

\author{
Ching-Lueh Chang\footnote{Department of Computer Science and
Engineering,
%\& Innovation Center for Big Data and Digital Convergence,
Yuan Ze University, Taoyuan, Taiwan.
clchang@saturn.yzu.edu.tw}
}

%\footnote{Supported
%in part by the Ministry of Science and Technology of Taiwan under
%grant 109-2221-E-155-031-.}
%\footnote{Innovation Center for Big Data and Digital Convergence, Yuan Ze University, Taoyuan, Taiwan.}
%\footnote{Supported in part by the Ministry of Science and Technology
%of Taiwan under
%grant
%101-2221-E-155-015-MY2.}
%103-2221-E-155-026-MY2.}
%}

%\institute{
%Department of Computer Science and Information Engineering, National Taiwan
%University, Taipei, Taiwan.\\
%\email{d95007@csie.ntu.edu.tw}
%\and
%Department of Computer Science and Information Engineering, National Taiwan
%University, Taipei, Taiwan.\\
%\email{lyuu@csie.ntu.edu.tw}
%}

%\author{Ching-Lueh Chang\thanks{Department of Computer Science and
%Information Engineering, National Taiwan University, Taipei,
%Taiwan. Email: d95007@csie.ntu.edu.tw.} \and Yuh-Dauh Lyuu\thanks{Department
%of Computer Science and Information
%Engineering, National Taiwan University, Taipei,
%Taiwan. Email: lyuu@csie.ntu.edu.tw.}}
\maketitle

\begin{abstract}
Given an $n$-point metric space $(M,d)$,
{\sc metric $1$-median} asks for a point $p\in M$ minimizing
$\sum_{x\in M}\,d(p,x)$.
We show that for each computable function $f\colon \mathbb{Z}^+\to\mathbb{Z}^+$
satisfying $f(n)=\omega(1)$,
{\sc metric $1$-median} has a deterministic, $o(n)$-query,
$o(f(n)\cdot\log n)$-approximation and nonadaptive algorithm.
Previously, no deterministic $o(n)$-query $o(n)$-approximation
algorithms are known for {\sc metric $1$-median}.
On the negative side,
we
%refute the existence of a
prove each
deterministic
$O(n)$-query
%$o(\log n)$-approximation
algorithm
for
{\sc metric $1$-median}
to be not $(\delta\log n)$-approximate for a sufficiently small constant $\delta>0$.
%that makes
%making
%each point
%in
%$\{0,1,\ldots,n-1\}$
%$M$
%involve
%in only $O(1)$ queries to $d$.
%As a corollary,
We also refute the existence of
%{\sc metric $1$-median}
%has no
deterministic $o(n)$-query $O(\log n)$-approximation algorithms.
\end{abstract}


\section{Introduction}

An $n$-point metric space $(M,d)$ is a size-$n$ set $M$ endowed with a distance
function $d\colon M\times M\to[0,\infty)$ such that
\begin{itemize}
\item $d(x,y)=0$ if and only if $x=y$,
\item $d(x,y)=d(y,x)$, and
\item $d(x,y)+d(y,z)\ge d(x,z)$ (triangle inequality)
\end{itemize}
for all $x$, $y$, $z\in M$~\cite{Rud76}.
{\sc Metric $1$-median} asks for a point $p\in M$
minimizing $\sum_{x\in M}\,d(p,x)$.
Clearly, it has a brute-force $O(n^2)$-time algorithm.
Furthermore, it generalizes the classical median selection~\cite{CLRS09}
and can be generalized further to metric $k$-median clustering.
In social network analysis, {\sc metric $1$-median}
asks for an actor with the maximum closeness centrality~\cite{WF94}.
For all $\beta\ge1$, a $\beta$-approximate $1$-median of $(M,d)$
is a point $p\in M$ satisfying
$\sum_{y\in M}\,d(p,y)\le \beta\cdot\min_{q\in M}\sum_{y\in M}\,d(q,y)$.
By convention, a $\beta$-approximation algorithm for
{\sc metric $1$-median} must output a
$\beta$-approximate $1$-median of $(M,d)$.
A query inspects $d(x,y)$ for some $x$, $y\in M$.
%An algorithm is nonadaptive if
%each query is independent of its previous queries.
%its sequence of queries depends only on $n$.
An algorithm is nonadaptive if
%Being nonadaptive means that the
its
$i$th query $(x_i,y_i)\in M^2$
%to $d$
is independent of the
answers to the
first $i-1$ queries, for all $i>1$.
%)
%For each undirected graph $G$,
Write $d_G$ for the distance function induced by an undirected graph $G$.
\comment{ % moved to the section that uses this 20211209 16:59
By convention, $d(x,S)\equiv \inf_{s\in S}\,d(x,s)$ for all $x\in M$ and $S\subseteq M$.
}% moved to the section that uses this 20211209 16:59

Indyk~\cite{Ind99,Ind00} gives a Monte Carlo $O(n/\epsilon^2)$-time
$(1+\epsilon)$-approximation
algorithm for {\sc metric $1$-median}, where $\epsilon>0$.
His time complexity is optimal w.r.t.\ $n$.
When restricted to $\mathbb{R}^D$,
{\sc metric $1$-median} has a Monte Carlo $O(D\cdot\exp(\text{poly}(1/\epsilon)))$-time
$(1+\epsilon)$-approximation algorithm~\cite{KSS10}.
%Given $S\subseteq M$ and $p\in M$,
The more general
$k$-median clustering in metric spaces has streaming approximation
algorithms~\cite{GMMMO03},
requires $\Omega(nk)$ time for $O(1)$-approximations~\cite{MP04}
and is inapproximable to within
$(1+2/e-\Omega(1))$ unless
$\text{NP}\subseteq \text{DTIME}(n^{O(\log\log n)})$~\cite{JMS02}.
For
%the Euclidean
$\mathbb{R}^D$
and graph metrics,
%A
a well-studied
%a
%related
problem is to find
%$\sum_{s\in S}\,d(p,s)$.
the average distance from a query point to a finite set of
points~\cite{BMM03,EW04,GR08}.
%It is studied for the Euclidean and graph metrics~\cite{BMM03,EW04,GR08}.
%Other extensively studied topics include streaming
%algorithms~\cite{GMMMO03} for and inapproximability of

%Because
%the brute-force exact algorithm for
%there
Deterministic
%$o(n^2)$-query
$\omega(n)$-query
%($o(n^2)$-query)
computation is almost completely understood
for {\sc metric $1$-median}:
For all constants $\epsilon\in(0,1)$,
%The
the
best approximation
ratio achievable by deterministic
$o(n^2)$-query
and
$O(n^{1+\epsilon})$-query algorithms is
$4$ and
$2\lceil 1/\epsilon\rceil$,
respectively~\cite{Cha17,Cha18,Wu14}.
%\begin{itemize}
%\item The best approximation
%ratio achievable by deterministic $o(n^2)$-time ($o(n^2)$-query, respectively)
%algorithms is $4$~\cite{Cha17}.
%\item For all constants $\epsilon\in(0,1)$, the
%best approximation
%ratio achievable by deterministic $O(n^{1+\epsilon})$-time
%($O(n^{1+\epsilon})$-query, respectively)
%algorithms is $2\lceil 1/\epsilon\rceil$~\cite{Cha18}.
%\end{itemize}
The same holds with ``query'' replaced by ``time'' and
regardless of whether the algorithms can be adaptive~\cite{Cha17,Cha18}.
%(An adaptive algorithm is allowed to)
In contrast, we
%will focus on the less understood
study the largely unknown
deterministic
$O(n)$- or
$o(n)$-query
computation.
%The
An
%query complexity of $o(n)$
%is special because
%has a special meaning:
%A strength of
%An
$o(n)$-query
%implies ignoring
algorithm
%must
enjoys the strength of ignoring
%algorithms
%is that they
%ignore
a $1-o(1)$
fraction of
%all $n$
%{\em points}.
%points.
points.
%must be ignored.
%This is in contrast with
\comment{ % does not add value to o(n)-query computation? 20211209 17:10
Instead,
%While
$o(n^2)$-query algorithms
%, which
ignore a $1-o(1)$
%inspect an $o(1)$
fraction of all $\binom{n}{2}$
distances.
}% does not add value to o(n)-query computation? 20211209 17:10
%$o(n)$-query algorithms
%must
%ignore a $1-o(1)$
%could inspect only an $o(1)$
%fraction of all $n$ {\em points}.
%So a query complexity of $o(n)$ has a special meaning.
%,
%where a point $x$
%is ignored if there is no $y\in M$
%for which $d(x,y)$ is queried.
% no $y\in M$
%(i.e., only $o(n)$ points $x$ are such that there exists $y$ for which
%$d(x,y)$ is queried).

It is folklore that every point is an $(n-1)$-approximate $1$-median.
Surprisingly, this is
%currently
the
current
best upper bound
%of all
for
deterministic $o(n)$-query
algorithms.
In particular, no deterministic $o(n)$-query $o(n)$-approximation algorithms
are known for {\sc metric $1$-median}.
\comment{ % moved to lower bound part 20211206 1:53
Currently, the best lower bound against deterministic $o(n)$-query
algorithms is that they cannot be $O(1)$-approximate; this remains
true with ``$o(n)$''
replaced by ``$O(n)$''~\cite{Cha18}.
}% moved to lower bound part 20211206 1:53
%In contrast, we
Instead, we
%We
give
a deterministic, $o(n)$-query,
$o(f(n)\cdot \log n)$-approximation and nonadaptive algorithm for each
computable function $f\colon \mathbb{Z}^+\to\mathbb{Z}^+$ satisfying
$f(n)=\omega(1)$.
So, e.g., {\sc metric $1$-median} has a deterministic $o(n)$-query
$o(\alpha(n)\cdot\log n)$-approximation algorithm for the very slowly growing
inverse Ackermann function $\alpha(\cdot)$.
\comment{ % said too many times 20211206 1:54
Previously, no
deterministic $o(n)$-query $o(n)$-approximation algorithms are known.
}% said too many times 20211206 1:54
Our main technical discovery
%, which may find other applications,
is that a $\beta$-approximate $1$-median of
$(S,d\vert_{S\times S})$ (where $d\vert_{S\times S}$ denotes
$d$ restricted to $S\times S$) is an $O(\beta n/|S|)$-approximate
$1$-median of $(M,d)$, for all $\emptyset\subsetneq S\subseteq M$ and $\beta\ge1$.
%We hope this new observation to find other applications.
%We do not know whether {\sc metric $1$-median} has a deterministic $o(n)$-query $o(\log n)$-approximation algorithm, though.
When $S\subseteq M$ is a uniformly random set of a sufficiently large size,
%Czumaj and Sohler~\cite{CS07}
%show that
an approximate solution to metric $k$-median clustering
for $(S,d\vert_{S\times S})$
%yields
is
%a good approximation
a
good
one
for $(M,d)$ with high probability~\cite{CS07}.
But our discovery is for {\em any}
%(\em deterministic)
$S$ and is new.

Chang~\cite{Cha17ICASI}
shows that {\sc metric $1$-median} has a deterministic,
$O(\exp(O(1/\epsilon))\cdot n\log n)$-time, $O(\exp(O(1/\epsilon))\cdot n)$-query,
$(\epsilon\log n)$-approximation and nonadaptive algorithm, for all $\epsilon>0$.
So deterministic $O(n)$-query algorithms can be $(\epsilon\log n)$-approximate
for each $\epsilon>0$.
%We prove that his approximation ratio of $\epsilon\log n$ cannot be improved
%to $o(\log n)$.
%In particular, we
Currently, the best lower bound against deterministic $O(n)$-query
algorithms is that they cannot be $O(1)$-approximate~\cite{Cha18}.
%, leaving
%a huge gap to be closed.
So there is a huge gap between
Chang's~\cite{Cha17ICASI}
%Chang
approximation ratio of $\epsilon\log n$ and the current best lower bound.
%On the other hand,
%we
%Our new lower bound shows the optimality of Chang's~\cite{Cha17ICASI}
%approximation ratio of $\epsilon\log n$.
We close the gap by
%refuting
showing
%For a lower bound,
%On the negative side,
%we
%show
%the
%nonexistence
%existence
%of
%Consider any
%no
%that
%a
each
deterministic
$O(n)$-query algorithm
for {\sc metric $1$-median}
to be not
$(\delta\log n)$-approximate
for a sufficiently small constant $\delta>0$ (depending on the algorithm).
%$o(\log n)$-approximation algorithm
%for {\sc metric $1$-median}.
Our approach, sketched below, adversarially answers the queries of a deterministic
$O(n)$-query algorithm {\sf Alg}:
\begin{enumerate}[(I)]
\item\label{startwithalledgesintuition} Start with the complete graph on $M$.
\item\label{markexpanderinitiallyintuition} Mark all edges in an $O(1)$-regular
expander graph as permanent.
\item Repeat the following:
  \begin{enumerate}[(1)]
  \item\label{answerwithshortestintuition} Upon receiving a query $(a,b)\in M^2$, find a shortest $a$-$b$ path $P$
  and answer by the length of $P$.
  \item\label{markaspermanentintuition} Mark all edges of $P$ as permanent.
  \item\label{keepingdegreessmall} For each vertex $v$ incident to too many permanent edges,
  remove all non-permanent edges incident to $v$.
  \end{enumerate}
\end{enumerate}
Intuitively,
item~(\ref{keepingdegreessmall}) keeps degrees small, thus forcing the output
of
%an algorithm
{\sf Alg} to have a large average distance to other points.
Because item~(\ref{answerwithshortestintuition}) answers a query by the length of $P$,
items~(\ref{markaspermanentintuition})--(\ref{keepingdegreessmall})
%should
%maintains
%consistency among queries
must preserve all edge of $P$ (by marking them as permanent and not removing
them) for the consistency in answering future queries.
Items~(\ref{startwithalledgesintuition})~and~(\ref{answerwithshortestintuition})--(\ref{keepingdegreessmall})
follow Chang's~\cite{Cha18} paradigm.
To prove a lower bound against {\sf Alg}, we shall make the output of {\sf Alg}
%bad enough
a lot worse than a $1$-median, presumably
by identifying or planting
%However, Chang fails to plant
a vertex with a sufficiently small average distance
to other points.
%; this is crucial in making the output of Alg bad enough.
However, Chang fails in this respect.
We overcome his problem by
%Our new proof idea is
item~(\ref{markexpanderinitiallyintuition}), which allows a vertex to have an
$O(1)$ average distance to other vertices.
%Otherwise, future queries may be answered inconsistently with the current one.
%that makes
%making
%each point involve in
%$O(1)$ queries.
%We show that
%cannot
%such an algorithm
%to
%be $o(\log n)$-appproximate.
%As in
\comment{ % should give an overview of Cha18, 20211205, 18:23
Like
%previous
existing
lower bounds,
%for {\sc metric $1$-median},
%our lower bound uses
we use
the adversarial method
(see~\cite{Cha18} and the references therein).
Nonetheless,
%Chang's~\cite{Cha18}
%it
our lower bound
%does not seem to be
is not
an easy corollary of any existing
result.
}% should give an overview of Cha18, 20211205, 18:23

An extension of our
%Our
%As a corollary to our
lower bound
%implies
forbids
each deterministic $o(n)$-query algorithm
%$A$
for
{\sc metric $1$-median}
%has no deterministic $o(n)$-query
%$O(\log n)$-approximation algorithms.
%fails
to be $o(f(n)\cdot\log n)$-approximate
for some computable function $f\colon \mathbb{Z}^+\to\mathbb{Z}^+$
satisfying
$f(n)=\omega(1)$.
%More generally,
In particular,
deterministic $o(n)$-query $O(\log n)$-approximation algorithms
do not exist.
%So
\comment{ % maybe too cumbersome here 20211226 18:36
Furthermore,
our
deterministic $o(n)$-query
$o(f(n)\cdot\log n)$-approximation algorithm (where $f(n)=\omega(1)$
is any computable function)
%is almost optimal.
cannot be improved.
}% maybe too cumbersome here 20211226 18:36
Previously, the best lower bound against deterministic
$o(n)$-query algorithms $A$
is folklore and
%says that
%is that they
forbids
$A$ to be
%asserts the existence of
%a function
%$h_A(n)=\omega(1)$
%such that $A$ is not
$h_A(n)$-approximate for some $h_A(n)=\omega(1)$.\footnote{For
a sketch of proof, answer all queries of $A$ by $1$ and put all points not
involved in the queries to be
extremely close to one another
but
extremely far away from $A$'s output and from the points involved in the queries.}
%but extremely close to one another.
So previous works
%we do not know
do not yet refute
%whether
the existence of
%Note that the
deterministic
$o(n)$-query
$O(\alpha(n))$-approximation
algorithms,
% can be $O(\alpha(n))$-approximate
%prior to this work,
where $\alpha(\cdot)$ is the very slowly growing inverse
Ackermann function.

Chang~\cite{Cha18ICS}'s adversarial method shows that {\sc metric $1$-median}
has no deterministic $O(n)$-query
$o(\log n)$-approximation
algorithms that make each point involve in
$O(1)$ queries to $d$.
But his adversary is rather na{\"\i}ve and does not seem to yield any
unconditional lower bound such as ours.

\section{Upper bound}

Take an $n$-point metric space $(M,d)$ and $\emptyset \subsetneq S\subseteq M$.
Define
\begin{eqnarray*}
x^*&\equiv&\mathop{\mathrm{argmin}}_{x\in M}\,\sum_{y\in M}\,d(x,y),\\
x^*_S&\equiv&\mathop{\mathrm{argmin}}_{x\in S}\,\sum_{y\in S}\,d(x,y)
\end{eqnarray*}
to be a $1$-median of $(M,d)$ and $(S,d\vert_{S\times S})$, respectively,
breaking ties arbitrarily.
Furthermore, pick $\bs{u}$ and $\bs{v}$ independently and uniformly
at random from $S$.
So
$$
\bar{r}_S\equiv \mathop{E}\left[\,d\left(\bs{u},\bs{v}\right)\,\right]
$$
is the average distance in $(S,d\vert_{S\times S})$.

\begin{lemma}\label{localoptimalquality}
$$\sum_{y\in S}\,d\left(x^*,y\right)\ge \frac{|S|\,\bar{r}_S}{2}.$$
\end{lemma}
\begin{proof}
We have
\begin{eqnarray*}
\sum_{y\in S}\,d\left(x^*,y\right)
&=&|S|\cdot \mathop{E}\left[\,d\left(x^*,\bs{u}\right)\,\right]\\
&=&\frac{1}{2}\cdot\left(
|S|\cdot \mathop{E}\left[\,d\left(x^*,\bs{u}\right)\,\right]
+|S|\cdot \mathop{E}\left[\,d\left(x^*,\bs{v}\right)\,\right]
\right)\\
&\ge& \frac{1}{2}\cdot
|S|\cdot \mathop{E}\left[\,d\left(\bs{u},\bs{v}\right)\,\right].
\end{eqnarray*}
%\qed
\end{proof}

\begin{lemma}\label{localoptimalupperbound}
$$
\sum_{y\in S}\, d\left(x^*_S,y\right)\le |S|\,\bar{r}_S.
$$
\end{lemma}
\begin{proof}
%We have
By the optimality of $x^*_S$,
$$
\sum_{y\in S}\, d\left(x^*_S,y\right)
\le \mathop{E}\left[\,\sum_{y\in S}\, d\left(\bs{u},y\right)\,\right].
$$
Clearly,
$$
\mathop{E}\left[\,\sum_{y\in S}\, d\left(\bs{u},y\right)\,\right]
=|S|\cdot
\mathop{E}\left[\,d\left(\bs{u},\bs{v}\right)\,\right].
$$
%where the inequality follows from the optimality of $x^*_S$.
%\qed
\end{proof}

%The next two lemmas consider
%Let $x'_S$ be
%any $\alpha$-approximate $1$-median $x'_S$ of $(S,d\vert_{S\times S})$,
%where $\alpha\ge1$.
%Clearly,
For all $x'_S\in S$,
\begin{eqnarray}
\sum_{y\in M}\,d\left(x'_S,y\right)
\le
\sum_{y\in M}\,\left(d\left(x'_S,x^*\right)+d\left(x^*,y\right)\right)
=
n\cdot d\left(x'_S,x^*\right)
+\sum_{y\in M}\,d\left(x^*,y\right).\label{dontknowhowtoname1}
\end{eqnarray}

The next two lemmas constitute our main discovery.

\begin{lemma}\label{localvsglobalclosertooptimal}
%If
For all $x'_S\in S$ and $\beta\ge1$ satisfying
$\sum_{y\in S}\,d(x'_S,y)\le \beta\cdot \sum_{y\in S}\,d(x^*_S,y)$
and
$d(x'_S,x^*)\le 2\beta \bar{r}_S$,
%where $\alpha\ge1$,
%then
$x'_S$ is an $O(\beta n/|S|)$-approximate $1$-median of $(M,d)$.
\end{lemma}
\begin{proof}
By Lemma~\ref{localoptimalquality},
\begin{eqnarray}
n\cdot d\left(x'_S,x^*\right)
\le
n\cdot d\left(x'_S,x^*\right)
\cdot\frac{2}{|S|\,\bar{r}_S}
\cdot\sum_{y\in S}\,d\left(x^*,y\right).
%\le
%n\cdot d\left(x'_S,x^*\right)
%\cdot\frac{2}{|S|\,\bar{r}_S}
%\cdot\sum_{y\in M}\,d\left(x^*,y\right).
\label{dontknowhowtoname2}
\end{eqnarray}
As $d(x'_S,x^*)\le 2\beta \bar{r}_S$ and $S\subseteq M$,
$$
\sum_{y\in M}\,d\left(x'_S,y\right)
\le O\left(\frac{\beta n}{|S|}\right)\cdot
\sum_{y\in M}\,d\left(x^*,y\right)
$$
by equations~(\ref{dontknowhowtoname1})--(\ref{dontknowhowtoname2}).
%\qed
\end{proof}

\begin{lemma}\label{localvsglobalfurtherfromoptimal}
%If
For all $x'_S\in S$ and $\beta\ge1$ satisfying
$\sum_{y\in S}\,d(x'_S,y)\le \beta\cdot \sum_{y\in S}\,d(x^*_S,y)$
and
$d(x'_S,x^*)>2\beta \bar{r}_S$,
%where $\alpha\ge1$,
%then
$x'_S$ is an $O(n/|S|)$-approximate $1$-median of $(M,d)$.
\end{lemma}
\begin{proof}
By the triangle inequality,
\begin{eqnarray}
\sum_{y\in S}\,d\left(x^*,y\right)
\ge \sum_{y\in S}\,\left(d\left(x'_S,x^*\right)-d\left(x'_S,y\right)\right)
=|S|\cdot d\left(x'_S,x^*\right)-\sum_{y\in S}\, d\left(x'_S,y\right).
\label{againdontknowhowtoname1}
\end{eqnarray}
%As $x'_S$ is $\alpha$-approximate in $(S,d\vert_{S\times S})$,
Furthermore,
\begin{eqnarray}
\sum_{y\in S}\, d\left(x'_S,y\right)
\le \beta\cdot \sum_{y\in S}\, d\left(x^*_S,y\right)
\stackrel{\text{Lemma~\ref{localoptimalupperbound}}}{\le}
\beta\, |S|\,\bar{r}_S.
\label{againdontknowhowtoname2}
\end{eqnarray}
%By equations~(\ref{againdontknowhowtoname1})--(\ref{againdontknowhowtoname2})
%and as
As
$d(x'_S,x^*)>2\beta \bar{r}_S$,
$$
\sum_{y\in S}\,d\left(x^*,y\right)
\stackrel{\text{(\ref{againdontknowhowtoname1})--(\ref{againdontknowhowtoname2})}}{\ge}
|S|\cdot d\left(x'_S,x^*\right)
-\beta\, |S|\,\bar{r}_S
> \frac{|S|}{2}\cdot d\left(x'_S,x^*\right).
$$
So
$$
n\cdot d\left(x'_S,x^*\right)
=\frac{2n}{|S|}\cdot \frac{|S|}{2}\cdot d\left(x'_S,x^*\right)
< \frac{2n}{|S|}\cdot \sum_{y\in S}\,d\left(x^*,y\right).
$$
This and equation~(\ref{dontknowhowtoname1}) imply
$$
\sum_{y\in M}\,d\left(x'_S,y\right)
\le O\left(\frac{n}{|S|}\right)\cdot
\sum_{y\in M}\,d\left(x^*,y\right).
$$
%\qed
\end{proof}

Lemmas~\ref{localvsglobalclosertooptimal}--\ref{localvsglobalfurtherfromoptimal}
%imply $x'_S$ to be a global $O(\alpha n/|S|)$-approximate $1$-median, as
%stated below.
imply the following.
%local vs.\ global

\begin{lemma}\label{localvsglobalratio}
For all $\beta\ge1$,
every $\beta$-approximate $1$-median of $(S,d\vert_{S\times S})$
is an
$O(\beta n/|S|)$-approximate $1$-median of $(M,d)$.
\end{lemma}

The following theorem is due to Chang~\cite{Cha17ICASI}.

\begin{theorem}[\cite{Cha17ICASI}]\label{squarequeryalgorithm}
For all constants $\epsilon>0$,
{\sc metric $1$-median} has a deterministic,
$O(\exp(O(1/\epsilon))\cdot n\log n)$-time,
$(\exp(O(1/\epsilon))\cdot n)$-query,
$O(\epsilon\cdot\log n)$-approximation and
nonadaptive algorithm.
%(Being nonadaptive means that the $i$th query $(x_i,y_i)\in M^2$
%to $d$ is independent of the first $i-1$ queries, for all $i>1$.)
\end{theorem}

Below is our main theorem.

\begin{theorem}\label{maintheorem}
For each computable function $f\colon\mathbb{Z}^+\to\mathbb{Z}^+$
satisfying
$f(n)=\omega(1)$,
%Then
{\sc metric $1$-median} has a deterministic,
$o(n)$-query, $o(f(n)\cdot \log n)$-approximation and
nonadaptive algorithm.
\end{theorem}
\begin{proof}
Take any
%set
$S\subseteq M$ of size $\Theta(n/\sqrt{f(n)})$.
Applying Theorem~\ref{squarequeryalgorithm} to $(S,d\vert_{S\times S})$,
an $O(\log |S|)$-approximate $1$-median $x'_S$ of $(S,d\vert_{S\times S})$
can be found deterministically and nonadaptively with
$O(|S|)$ queries.
By Lemma~\ref{localvsglobalratio} (with $\beta=O(\log |S|)$),
$x'_S$ is an $O((\log |S|)\cdot n/|S|)$-approximate $1$-median of $(M,d)$.
%\qed
\end{proof}

Taking a very slowly growing $f(\cdot)$ (e.g., the iterated
logarithm
%$\log^*n$
or the inverse
Ackermann function),
Theorem~\ref{maintheorem} allows deterministic $o(n)$-query algorithms
to be very close to being $O(\log n)$-approximate.
%E.g., take $f(\cdot)$ to be slowly growing functions such as

\section{Lower bound}

Fix any deterministic $q$-query algorithm {\sf Alg}, where $q=q(n)=O(n)$.
Then take a constant $C>2d+4q/n$, where $d=O(1)$ is
%a constant
such that $d$-regular expander graphs exist.
By padding, assume the number of {\sf Alg}'s queries to be exactly $q$.
Adversary {\sf Adv} in Fig.~\ref{adversary} answers the queries of {\sf Alg}.
All graphs are assumed to be undirected.

\begin{figure}
\begin{algorithmic}[1]
\STATE Let $G^{(0)}$ be the complete graph on $M$;
\STATE Pick a $d$-regular expander graph $G^{\text{exp}}$
%=(M,E^{\text{Exp}})$
on $M$,
where
%$d\in\mathbb{Z}^+$ is a constant;
$d=O(1)$;
%\FOR{each edge $e\in E^{\text{Exp}}$
\STATE Mark all edges of $G^{\text{exp}}$ as permanent;
\FOR{$i=1$ up to $q$}
  \STATE Receive the $i$th query, denoted by $(a_i,b_i)\in M^2$;
  \STATE Pick
  %Let $P_i$ be
  a shortest $a_i$-$b_i$ path $P_i$ in $G^{(i-1)}$;
  \STATE Answer the $i$th query by the length of $P_i$;
  \STATE Mark all edges of $P_i$ as permanent;
  \STATE $G^{(i)}\leftarrow G^{(i-1)}$;
  \FOR{each $v\in M$}
    \IF{$v$ is incident to more than $C$ permanent edges}
      \STATE Remove from $G^{(i)}$ all non-permanent edges incident to $v$;
    \ENDIF
  \ENDFOR
\ENDFOR
\end{algorithmic}
\caption{Adversary {\sf Adv} for answering the queries of {\sf Alg}}
\label{adversary}
\end{figure}

As a remark,
whenever
an edge of a graph
%(such as $G^{\text{\rm Exp}}$)
is marked as permanent,
that edge is considered to be permanent in all graphs.
% (such as $G^{(i-1)}$).}
For example, an edge of $G^{\text{exp}}$ marked as permanent
in line~3 of {\sf Adv}
%(in Fig.~\ref{adversary})
is considered to be permanent in lines~11--13, even though the latter
processes $G^{(i)}$ rather than $G^{\text{exp}}$.
%For another example,
Similarly,
although
an edge marked as permanent by
line~8 comes from $G^{(i-1)}$ by line~6,
%but
%But
it is considered to be permanent
in lines~11--13 as well.
%(that process $G^{(i)}$).
\comment{ % seems unnecessarily detailed 20211212 20:12
The condition in line~11 is that more than $C$ edges
with $v$ as an endpoint
have been marked as permanent, no matter the marking was on $G^{\text{exp}}$
in line~3 or on $P_i$ (for some $i$) in line~8.
}% seems unnecessarily detailed 20211212 20:12

\begin{lemma}\label{expanderisembedded}
For all $0\le i\le q$,
$G^{\text{\rm exp}}$ is a subgraph of $G^{(i)}$.
%and line~6
%and
%of {\sf Adv} (in Fig.~\ref{adversary}) is valid (i.e.,
%$G^{(i-1)}$
%has an $a_i$-$b_i$ path (so line~6 of {\sf Adv} in Fig.~\ref{adversary} is valid).
\end{lemma}
\begin{proof}
%By line~1 of {\sf Adv} (in Fig.~\ref{adversary}),
%$G^{\text{\rm Exp}}$ is a subgraph of $G^{(0)}$.
%Line~3 marks all edges of $G^{\text{\rm Exp}}$ as permanent.
%Finally, line~12 does not remove permanent edges.
%\comment{ % too complicated, I think 20211205 20:46
By line~1,
% of {\sf Adv} (in Fig.~\ref{adversary}),
$G^{\text{\rm exp}}$ is a subgraph of $G^{(0)}$.
%As $G^{\text{\rm Exp}}$ is an expander, $G^{(0)}$ is connected.
Assume as induction hypothesis that $G^{\text{\rm exp}}$
is a subgraph of $G^{(i-1)}$.
%appears in $G^{(i)}$ if it.
%As all edges of $G^{\text{\rm Exp}}$ are permanent by line~3,
By line~3 and
the induction hypothesis,
%As $G^{\text{\rm Exp}}$
%is a subgraph of $G^{(i-1)}$, every line~3.
%implies that
all edges of $G^{\text{\rm exp}}$ are permanent edges of $G^{(i-1)}$.
%\footnote{Whenever
%an edge of a graph (such as $G^{\text{\rm Exp}}$) is marked as permanent,
%that edge is considered to be permanent in all graphs (such as $G^{(i-1)}$).}
By
lines~9--14, all permanent edges of $G^{(i-1)}$ are in $G^{(i)}$.
%}% too complicated, I think 20211205 20:46
\end{proof}

\begin{lemma}[{Implicit in~\cite{Cha18}}]\label{consistencywiththefinaldistance}
For all $1\le i\le q$,
{\sf Adv}'s answer to the $i$th query of {\sf Alg} equals $d_{G^{(q)}}(a_i,b_i)$.
\end{lemma}
\begin{proof}[Proof (included for completeness)]
Let ${\text{ans}}_i$ be {\sf Adv}'s answer to the $i$th query.
By lines~6--7, ${\text{ans}}_i=d_{G^{(i-1)}}(a_i,b_i)$.\footnote{As $G^{\text{exp}}$
is an expander, $d_{G^{(i-1)}}(a_i,b_i)<\infty$ by Lemma~\ref{expanderisembedded}.}
By lines~9--14, $G^{(q)}$ is a subgraph of $G^{(i-1)}$, implying $d_{G^{(i-1)}}(a_i,b_i)
\le d_{G^{(q)}}(a_i,b_i)$.
In summary, ${\text{ans}}_i\le d_{G^{(q)}}(a_i,b_i)$.

By line~7, ${\text{ans}}_i$ is the length of $P_i$.
%In the $i$th iteration of the loop in lines~--
As $P_i$ is in $G^{(i-1)}$ by line~6, all edges of $P_i$ are permanent edges
of $G^{(i)}$ by lines~8--14.
%By lines~6~and~8--9
%Lines~8 marks as permanent all edges of $P_i$.
%Line~9 initializes $G^{(i)}$
%of $P_{i-1}$ as permanent in line~
So by lines~9--14, $P_i$ exists in $G^{(j)}$ for all $j\ge i$.\footnote{Note that
once an edge is marked as permanent, it cannot be removed by line~12.}
Therefore,
the length of $P_i$ is at least
$d_{G^{(q)}}(a_i,b_i)$
(in fact, at least $d_{G^{(j)}}(a_i,b_i)$ for all $j\ge i$).
In summary,
${\text{ans}}_i\ge d_{G^{(q)}}(a_i,b_i)$.
\end{proof}

\begin{lemma}[{Implicit in~\cite{Cha18}}]\label{permanentedgesgrowslowlyonavertex}
For each $v\in M$,
%an iteration in lines~
each run of line~8
marks as permanent at most two edges incident to $v$.
\end{lemma}
\begin{proof}[Proof (included for completeness)]
In line~6, $P_i$ has at most two edges incident to $v$.
%Then line~8 marks the edges of $P_i$ as permanent.
\end{proof}

Let $E^{\text{perm}}$ be the set of edges
ever
marked as permanent,
and $G^{\text{perm}}=(M,E^{\text{perm}})$.
Denote by $z^*\in M$ the output of {\sf Alg} with all queries answered by {\sf Adv}.
By padding dummy queries,
assume without loss of generality that {\sf Alg} queries for the distance
between $z^*$ and each point in $M$.
%E.g., if Alg terminates without querying
%for $d(z^*,x)$ for some $x\in M$, then add $()$ dummy queries.
%prior to termination.

\begin{lemma}[{Implicit in~\cite{Cha18}}]\label{heavilyqueriedpointsarebad}
$$\sum_{x\in M}\, d_{G^{(q)}}(z^*,x)=\Omega(n\log n).$$
\end{lemma}
\begin{proof}[Proof (included for completeness)]
By lines~7--8, {\sf Adv} answers each query of {\sf Alg}
%is answered
by the length of a path
%with edges only
whose edges are all
in
$E^{\text{perm}}$.
%$G^{\text{perm}}$.
So for all $i\ge1$,
the answer to the $i$th query is at least $d_{G^{\text{perm}}}(a_i,b_i)$.
%, where $i\ge 1$.
Therefore,
$d_{G^{(q)}}(a_i,b_i)\ge d_{G^{\text{perm}}}(a_i,b_i)$
by Lemma~\ref{consistencywiththefinaldistance}, where $i\ge 1$.
This and the assumption that
%As
{\sf Alg} queries for
all distances between $z^*$ and the points in $M$ give
\begin{eqnarray}
\sum_{x\in M}\, d_{G^{(q)}}(z^*,x)\ge\sum_{x\in M}\, d_{G^{\text{perm}}}(z^*,x).
\label{transforminganswerstothoseonpermanentgraph}
\end{eqnarray}

%By Lemma~\ref{permanentedgesgrowslowlyonavertex},
Consider the
%time step
instant
$t$ when the number of permanent
edges incident to a vertex $v\in M$ exceeds $C$.
By Lemma~\ref{permanentedgesgrowslowlyonavertex},
$v$ is incident to at most $C+2$ permanent edges at time $t$.
Then lines~9--14 remove
from $G^{(i)}$ all non-permanent edges incident to $v$ (and will not put them
back to $G^{(j)}$ for any $j>i$).
So no more edges incident to $v$ will be marked as permanent after time $t$.
In summary,
$v$ has degree at most $C+2$ in $G^{\text{perm}}$.
In the above argument, $v$ can be any vertex
whose number of incident permanent edges ever exceeds $C$.
So
$G^{\text{perm}}$ has maximum degree at most $C+2$.\footnote{Clearly,
a vertex
whose number of incident permanent edges never exceeds $C$
will have degree $\le C$ in $G^{\text{perm}}$.}
So for
all
$k\ge 1$,
at most $\sum_{h=0}^k\,(C+2)^h$ vertices in $G^{\text{perm}}$
%are
can be
within distance $k$ (inclusive) from $z^*$.
Taking $k=\epsilon \log n$ for a small constant $\epsilon>0$ depending on $C$,
$\sum_{h=0}^k\,(C+2)^h\le\sqrt{n}$.
I.e.,
at least $n-\sqrt{n}$
vertices
are of distance greater than $\epsilon\log n$ from $z^*$ in $G^{\text{perm}}$.
So
$$\sum_{x\in M}\, d_{G^{\text{perm}}}(z^*,x)\ge \left(n-\sqrt{n}\right)
\cdot \epsilon\log n.$$
This and inequality~(\ref{transforminganswerstothoseonpermanentgraph})
complete the proof.
%and Lemma~\ref{permanentedgesgrowslowlyonavertex}, each vertex
%has at most $C+2$ incident edges ever marked as permanent.\footnote{When more than $C$ edges
%incident to a vertex $v$ are permanent, at most $C+2$ edges incident to $v$ are 
%lines~--, it is removed from $G^{(i)}$}
\end{proof}

Let
$\text{Bad}\subseteq M$ be the set of vertices
with degrees
%incident to
at least
$C$
%edges
%ever marked as permanent.
in $G^{\text{perm}}$.

\begin{lemma}[{Implicit in~\cite{Cha18}}]\label{edgesbetweengoodverticesarepreserved}
For all distinct $y$, $z\in M\setminus\text{\rm Bad}$,
$d_{G^{(q)}}(y,z)=1$.
\end{lemma}
\begin{proof}[Proof (included for completeness)]
%Clearly, it suffices to show that $(y,z)$ is an edge of $G^{(q(n))%)}$ for all
%$z\in M\setminus\text{\rm Bad}$.
%Indeed, any two vertices in $M\setminus\text{\rm Bad}$
By line~1, $(y,z)$ is an edge of $G^{(0)}$.
As $y$, $z\notin \text{\rm Bad}$,
$y$ and $z$
%have degrees
are incident to fewer
%less
than $C$
edges
ever marked as permanent.
%in $G^{\text{perm}}$.
So lines~9--14
%always
preserve
%never remove
the edge $(y,z)$
in $G^{(i)}$
%for all $i\ge 1$.
%This and the fact that $G^{(0)}$ is complete
%Consequently,
%In summary, $G^{(i)}$
%$G^{(q(n)))}$
%has the edge
%has
%$(y,z)$
for all $i\ge 1$.
\end{proof}

By convention, $d(x,S)\equiv \inf_{s\in S}\,d(x,s)$
for all $x\in M$ and $S\subseteq M$.

\begin{corollary}\label{distancestogoodpoints}
For all $y\in M\setminus\text{\rm Bad}$,
$$\sum_{x\in M}\,d_{G^{(q)}}(x,y)
\le \sum_{x\in M}\,\left(d_{G^{(q)}}(x,M\setminus\text{\rm Bad})+1\right).$$
\end{corollary}
\begin{proof}
Assume $M\setminus\text{Bad}\neq \emptyset$
%for, otherwise, the corollary is vacuously true.
to avoid vacuous truth.
For each $x\in M$, let $z_x\in M\setminus\text{\rm Bad}$
satisfy
$$
d_{G^{(q)}}(x,M\setminus\text{\rm Bad})
=d_{G^{(q)}}(x,z_x).
$$
By Lemma~\ref{edgesbetweengoodverticesarepreserved},
$d_{G^{(q)}}(y,z_x)\le 1
$
for all $x\in M$.
By the triangle inequality,
$$d_{G^{(q)}}(x,y)\le d_{G^{(q)}}(x,z_x)+d_{G^{(q)}}(y,z_x),$$
where $x\in M$.
\end{proof}

\begin{lemma}[{Implicit in~\cite{Cha18}}]\label{shortestpathhasonenewnonpermanent}
For all $1\le i\le q$ and when line~6 picks $P_i$,
$P_i$ has at most one non-permanent edge.
\end{lemma}
\begin{proof}[Proof (included for completeness)]
Write $P_i=(v_1,v_2,\ldots,v_t)$.
Assume for contradiction that $(v_h,v_{h+1})$ and $(v_k,v_{k+1})$
are both non-permanent when line~6 picks $P_i$
from $G^{(i-1)}$,
for some $1\le h<k<t$.
By line~1, $G^{(0)}$ has the edge $(v_h,v_{k+1})$.
But by the optimality of $P_i$ in line~6,
$G^{(i-1)}$ cannot have the edge $(v_h,v_{k+1})$.
So there exists $1\le \ell\le i-1$ such that
line~12
%removes $(v_h,v_{k+1})$
runs with $v\in \{v_h,v_{k+1}\}$
in
the $\ell$th iteration of the loop in lines~4--15.\footnote{Let
$\ell$ be the smallest index such that $G^{(\ell)}$ does not have
$(v_h,v_{k+1})$.
Line~9 initializes $G^{(\ell)}$ to be $G^{(\ell-1)}$, which has
$(v_h,v_{k+1})$.
So
line~12 must remove
$(v_h,v_{k+1})$
%may be removed
from $G^{(\ell)}$.
This happens only
%only
by running line~12 with $v\in \{v_h,v_{k+1}\}$.}
%Consequently,
%line~12 must run with $v\in \{v_h,v_{k+1}\}$
Being non-permanent when line~6 picks $P_i$
from $G^{(i-1)}$,
$(v_h,v_{h+1})$ and $(v_k,v_{k+1})$
must
%be
have
remained
%to be
non-permanent
throughout the first $i-1$
iterations (including the $\ell$th iteration)
of the loop in lines~4--15 (because of the irreversibility of permanence).
Therefore,
%the first run of
when
line~12
%removes $(v_h,v_{k+1})$
runs
with $v\in \{v_h,v_{k+1}\}$
in
the $\ell$th iteration of the loop in lines~4--15,
%must remove
%it removes
%either
$(v_h,v_{h+1})$ or $(v_k,v_{k+1})$ must be removed from $G^{(\ell)}$.
%is removed.
By symmetry, assume $G^{(\ell)}$ to not have $(v_h,v_{h+1})$.
By lines~9--14 and as $\ell\le i - 1$,
%As $\ell\le i - 1$ and by lines~9--,
$G^{(i-1)}$ cannot have $(v_h,v_{h+1})$, either.
%But then $G^{(\ell)}$, $G^{(\ell+1)}$, $\ldots$ will not have the edge
%$(v_h,v_{h+1})$.
%In particular,
%$G^{(i-1)}$ does not have the edge
%$(v_h,v_{h+1})$.
%However, $(v_h,v_{h+1})$ is on $P_i$ and
As
$P_i$ is picked from $G^{(i-1)}$ by line~6,
%forcing
$G^{(i-1)}$
%to
must
have $(v_h,v_{h+1})$ (which is on $P_i$), a contradiction.
% rewritten somewhat 20211209 19:26
\end{proof}

\begin{corollary}[{Implicit in~\cite{Cha18}}]\label{numberofedgesmarkedaspermanent}
%Each iteration of the loop in lines~--
%marks at most one edge as permanent.
Each
%execution
run
of line~8 increases
the number of permanent edges
%increases
by at most one.
\end{corollary}
\begin{proof}[Proof (included for completeness)]
Immediate from Lemma~\ref{shortestpathhasonenewnonpermanent}.
\end{proof}

\begin{lemma}\label{fewbadpoints}
$|\text{\rm Bad}|\le n/2$.
\end{lemma}
\begin{proof}
As $G^{\text{exp}}$ is $d$-regular by line~2, line~3 marks
%at most
$dn/2$ edges as permanent by the handshaking lemma.
By Corollary~\ref{numberofedgesmarkedaspermanent},
%each run of
%line~8
%marks
at most
%one edge
$q$ edges
are ever marked
as permanent by line~8.
To sum up,
$G^{\text{perm}}$
has
at most $dn/2+q$ edges.
So by the handshaking lemma, the average degree in $G^{\text{perm}}$
is at most $d+2q/n$.
%are ever marked as permanent.
This and Markov's inequality imply that
%So by averaging (over the vertices),
at most $n/2$ vertices have degrees
%can be incident to
at least $2d+4q/n$
%permanent
in $G^{\text{perm}}$.
As $C>2d+4q/n$, at most $n/2$ vertices have degrees
%can be incident to
at least $C$
%permanent
in $G^{\text{perm}}$.
\end{proof}

\begin{lemma}\label{goodpointsarereallygood}
For all $y\in M\setminus\text{\rm Bad}$,
$\sum_{x\in M}\,d_{G^{(q)}}(x,y)=O(n)$.
\end{lemma}
\begin{proof}
By Lemmas~\ref{fewbadpoints}~and~\ref{averagedistanceleavingaset} (in
Appendix~\ref{expandersaveragedistanceleavingaset}),
\begin{eqnarray}
\sum_{x\in \text{Bad}}\, d_{G^{\text{exp}}}\left(x,M\setminus\text{Bad}\right)=O(n).
%\label{totaldistanceintobadset}
\nonumber
\end{eqnarray}
%By lines~3~and~--,
%All edges in $G^{\text{Exp}}$ are marked as permanent by line~3 and
%are preserved in $G^{(i)}$ by lines~-- for all $i\ge 1$.
%So $G^{\text{Exp}}$ is a subgraph of $G^{(q(n))}$, implying
%By Lemma~\ref{expanderisembedded},
%$$
%\sum_{x\in \text{Bad}}\, d_{G^{(q(n))}}\left(x,M\setminus\text{Bad}\right)
%\le
%\sum_{x\in \text{Bad}}\, d_{G^{\text{Exp}}}\left(x,M\setminus\text{Bad}\right).
%$$
This and
Lemma~\ref{expanderisembedded}
%equation~(\ref{totaldistanceintobadset})
give
\begin{eqnarray}
\sum_{x\in \text{Bad}}\, d_{G^{(q)}}\left(x,M\setminus\text{Bad}\right)
=O(n).
\label{sumofdistancesfrombadout}
\end{eqnarray}
%By Lemma~\ref{edgesbetweengoodverticesarepreserved},
Clearly,
\begin{eqnarray}
\sum_{x\in M\setminus\text{Bad}}\, d_{G^{(q)}}\left(x,M\setminus\text{Bad}\right)
\le\sum_{x\in M\setminus\text{Bad}}\, d_{G^{(q)}}\left(x,x\right)
=0.\label{sumofditancesfromgoodinside}
\end{eqnarray}
Now sum
%Sum
up equations~(\ref{sumofdistancesfrombadout})--(\ref{sumofditancesfromgoodinside})
and
%Finally,
invoke Corollary~\ref{distancestogoodpoints}.
\end{proof}

\begin{theorem}
%For
%each
Each
deterministic $O(n)$-query algorithm
%Alg
for
{\sc metric $1$-median}
%, there exists a constant $\delta>0$
%such that Alg
is not $(\delta\log n)$-approximate for a sufficiently small constant $\delta>0$.
%{\sc Metric $1$-median} has no deterministic $O(n)$-query $o(\log n)$-approximation algorithms.
\end{theorem}
\begin{proof}
By Lemma~\ref{consistencywiththefinaldistance}, {\sf Adv} answers
%queries
consistently with
$d_{G^{(q)}}(\cdot,\cdot)$.
By Lemmas~\ref{heavilyqueriedpointsarebad}~and~\ref{fewbadpoints}--\ref{goodpointsarereallygood},
{\sf Alg}'s output, $z^*$, satisfies
$$\sum_{x\in M}\,d_{G^{(q)}}(z^*,x)=\Omega(\log n)\cdot
\sum_{x\in M}\,d_{G^{(q)}}(y,x)$$
%is not an $o(\log n)$-approximate
%$1$-median w.r.t.\
%$d_{G^{(q(n))}}(\cdot,\cdot)$.
for some $y\in M$.
%for a sufficiently small constant $\delta>0$.
Finally, recall that {\sf Alg} is an arbitrary deterministic
$O(n)$-query algorithm.
\end{proof}

\subsection{Even fewer queries}

For all $n\in\mathbb{Z}^+$, $[n]\equiv\{1,2,\ldots,n\}$.
This subsection assumes $q=o(n)$ and $M=[n]$.
An algorithm is said to be tame if
its queries are in $[2q+1]\times [2q+1]$
and its output in $[2q+1]$.
%\begin{enumerate}[(A)]
%\item its queries are all in $[2q(n)]\times [2q(n)]$, and
%\item it queries for all distances between its output $z^*$
%and other points.
%\end{enumerate}

\begin{figure}
\begin{algorithmic}[1]
%\FOR{$i=1$ up to $n$}
%  \STATE $\pi(i)\leftarrow \bot$;
%\ENDFOR
\STATE $\text{cnt}\leftarrow 0$;
\FOR{$i=1$ up to $q$}
  \STATE Receive the $i$th query of {\sf Alg}, denoted by $(a_i,b_i)\in M^2$;
  \IF{$a_i\notin \{a_1,b_1,a_2,b_2,\ldots,a_{i-1},b_{i-1}\}$}
    \STATE $\text{cnt}\leftarrow \text{cnt}+1$;
    \STATE $\pi(a_i)\leftarrow \text{cnt}$;
  \ENDIF
  \IF{$b_i\notin \{a_1,b_1,a_2,b_2,\ldots,a_{i-1},b_{i-1}\}\cup\{a_i\}$}
    \STATE $\text{cnt}\leftarrow \text{cnt}+1$;
    \STATE $\pi(b_i)\leftarrow \text{cnt}$;
  \ENDIF
  \STATE Query for the distance between $\pi(a_i)$ and $\pi(b_i)$, and return
  the answer to {\sf Alg};
\ENDFOR
\STATE Receive the output $z^*$ of {\sf Alg};
\IF{$z^*\notin \{a_1,b_1,a_2,b_2,\ldots,a_{q},b_{q}\}$}
  \STATE $\text{cnt}\leftarrow \text{cnt}+1$;
  \STATE $\pi(z^*)\leftarrow \text{cnt}$;
\ENDIF
\RETURN $\pi(z^*)$;
\end{algorithmic}
\caption{Algorithm {\sf Sim} for simulating {\sf Alg} with points renamed}
\label{simulatorfororiginalalgorithmic}
\end{figure}

\begin{lemma}\label{injectiverenaming}
When {\sf Sim} (in Fig.~\ref{simulatorfororiginalalgorithmic}) terminates,
$\pi(\cdot)$ is injective.
%on $\{a_1,b_1,a_2,b_2,\ldots,b_{q(n)},b_{q(n)}\}\cup\{z^*\}$.
\end{lemma}
\begin{proof}
Before lines~6,~10~and~17, cnt increments.
\end{proof}

\begin{lemma}\label{smallrangeofpointsconcerned}
When {\sf Sim} terminates,
$\pi(a_i)$, $\pi(b_i)$, $\pi(z^*)\in[2q+1]$ for all $1\le i\le q$.
\end{lemma}
\begin{proof}
Each query increases cnt by at most two in lines~4--11.
Lines~15--18 may also increase cnt.
Lines~6,~10,~and~17 set $\pi(x)$ to be cnt for some $x\in M$.
\end{proof}

\begin{lemma}\label{simulationdoesnotchangeproperties}
If {\sf Alg} is $h(n)$-approximate for {\sc metric $1$-median}, where
$h\colon\mathbb{Z}^+\to\mathbb{R}$, then
{\sf Sim} is a tame $q$-query
$h(n)$-approximation algorithm for {\sc metric $1$-median}.
%with groundset restricted to $[n]$.
\end{lemma}
\begin{proof}
By Lemma~\ref{injectiverenaming}, {\sf Sim} simulates {\sf Alg} with an injective
renaming of points.
So, inheriting from {\sf Alg}, {\sf Sim} is $h(n)$-approximate and makes $q$ queries.
By Lemma~\ref{smallrangeofpointsconcerned} and lines~12~and~19 of {\sf Sim},
{\sf Sim} is tame.
%Clearly, {\sf Sim} makes translates each query of Alg into a new query.
\end{proof}

The following result complements Theorems~\ref{maintheorem}.

\begin{theorem}\label{lowerboundextendedtoverysmallquerycomplexity}
Each deterministic $o(n)$-query algorithm for
{\sc Metric $1$-median} fails to be
%{\sc Metric $1$-median} has no deterministic $o(n)$-query $O(\log n)$-approximation
%algorithms.
$o(f(n)\cdot\log n)$-approximate for some computable function
$f\colon\mathbb{Z}^+\to\mathbb{Z}^+$ satisfying
$f(n)=\omega(1)$.
\end{theorem}
\begin{proof}
By Lemma~\ref{simulationdoesnotchangeproperties}, assume {\sf Alg} to be tame
%and the groundset to be $[n]$
without loss of generality (otherwise, prove the theorem
against {\sf Sim} instead
of {\sf Alg}).
Let $z^*$ the {\sf Alg}'s output when the queries are answered by
%Answer {\sf Alg}'s queries by
{\sf Adv}
%and obtain {\sf Alg}'s output $z^*$,
with
$M$ (resp., $n$)
substituted by $[2q+1]$ (resp., $2q+1$).
By Lemma~\ref{heavilyqueriedpointsarebad} with $M$ (resp., $n$)
substituted by $[2q+1]$ (resp., $2q+1$),
\begin{eqnarray}
\sum_{x\in[2q+1]}\,d_{G^{(q)}}(z^*,x)=\Omega\left((2q+1)\log(2q+1)\right),
\label{badsolutioninlocalset}
\end{eqnarray}
where
%(Here
$G^{(q)}$ is a graph on $[2q+1]$ as in
%because
{\sf Adv}.
By Lemmas~\ref{fewbadpoints}--\ref{goodpointsarereallygood}
with $M$ (resp., $n$) substituted by $[2q+1]$ (resp., $2q+1$),
there exists $y\in[2q+1]$ satisfying
\begin{eqnarray}
\sum_{x\in[2q+1]}\,d_{G^{(q)}}(y,x)=O(q).
\label{bestpointbehaveswelllocally}
\end{eqnarray}
Equations~(\ref{badsolutioninlocalset})--(\ref{bestpointbehaveswelllocally})
and the triangle inequality
imply
\begin{eqnarray}
d_{G^{(q)}}(z^*,y)=\Omega(\log q).
\label{localsolutionfarawayfromlocaloptimal}
\end{eqnarray}

%Although distinct points should be a positive distance away, pretend
%as if the points in $[n]\setminus[2q+1]$ are copies of $y$.
Recall that $y\in[2q+1]$.
Put all points in $[n]\setminus[2q+1]$ extremely close to $y$:
%I.e., for
For
all distinct $a$, $b\in [n]$,
$d(a,a)\equiv 0$ and
\begin{eqnarray}
d(a,b)\equiv
\left\{
\begin{array}{ll}
%0, &\text{if $a$, $b\in[n]\setminus[2q+1]$,}\\
%d_{G^{(q)}}(a,y), &\text{if $a\in[2q+1]$ and $b\in[n]\setminus[2q+1]$,}\\
%d_{G^{(q)}}(y,b), &\text{if $a\in[n]\setminus[2q+1]$ and $b\in[2q+1]$,}\\
%d_{G^{(q)}}(a,b), &\text{otherwise.}
1/2^n, &\text{if $a$, $b\in\{y\}\cup([n]\setminus[2q+1])$,}\\
d_{G^{(q)}}(a,y), &\text{if $a\notin\{y\}\cup([n]\setminus[2q+1])$ and $b\in\{y\}\cup([n]\setminus[2q+1])$,}\\
d_{G^{(q)}}(y,b), &\text{if $a\in\{y\}\cup([n]\setminus[2q+1])$ and $b\notin\{y\}\cup([n]\setminus[2q+1])$,}\\
%d_{G^{(q)}}(y,b), &\text{if $a\in[n]\setminus[2q+1]$ and $b\in[2q+1]$,}\\
d_{G^{(q)}}(a,b), &\text{otherwise.}
\end{array}
\right.
\label{distancefunctionwithcopies}
\end{eqnarray}
It is not hard to see that $d$ is induced by the weighted graph obtained in the following way:
(1)~Add all vertices in $[n]\setminus[2q+1]$ to $G^{(q)}$.
(2)~Add an edge between each $v\in [n]\setminus[2q+1]$ and each neighbor (in $G^{(q)}$)
of $y$.
(3)~Connect any two vertices in
$\{y\}\cup([n]\setminus[2q+1])$ by an edge of weight $1/2^n$, all other edge weights being $1$.

%Without loss of generality, assume $(y,y)$ to not be a query of {\sf Alg}.
%So $d(a,b)=d(y,b)$ for all $a\in [n]\setminus[2q+1]$ and $b\in [n]$.
As {\sf Alg} is tame, $(a_i,b_i)\in[2q+1]\times[2q+1]$ for all $1\le i\le q$,
implying $d(a_i,b_i)=d_{G^{(q)}}(a_i,b_i)$ by
equation~(\ref{distancefunctionwithcopies}).
%\footnote{Without loss generality,
%assume $a_i\neq b_i$. So it is impossible that $a_i=b_i=y$.}
So by Lemma~\ref{consistencywiththefinaldistance},
{\sf Adv} answers queries consistently with $d(\cdot,\cdot)$.

%Recall that $y\in[2q+1]$.
%Then by
%By
%equations~(\ref{bestpointbehaveswelllocally})~and~(\ref{distancefunctionwithcopies}),
We have
\begin{eqnarray}
\sum_{x\in [n]\setminus\{y\}}\,d(y,x)
&=&\sum_{x\in [2q+1]\setminus\{y\}}\,d(y,x)
+\sum_{x\in [n]\setminus ([2q+1])\cup\{y\})}\,d(y,x)
\label{boundinggoodtoallfirst}\\
&\stackrel{\text{(\ref{distancefunctionwithcopies})}}{=}&
\sum_{x\in [2q+1]\setminus\{y\}}\,d(y,x)
+\sum_{x\in [n]\setminus ([2q+1])\cup\{y\})}\,\frac{1}{2^n}
\nonumber\\
&\stackrel{\text{(\ref{distancefunctionwithcopies})}}{=}&
\sum_{x\in [2q+1]\setminus\{y\}}\,d_{G^{(q)}}(y,x)
+\sum_{x\in [n]\setminus ([2q+1])\cup\{y\})}\,\frac{1}{2^n}
\nonumber\\
&\stackrel{\text{(\ref{bestpointbehaveswelllocally})}}{=}&
O(q).\label{boundinggoodtoalllast}
%\label{localoptimalverygoodglobally}
\end{eqnarray}
%By equation~(\ref{localsolutionfarawayfromlocaloptimal}),
As {\sf Alg} is tame, $z^*\in[2q+1]$.
By equation~(\ref{localsolutionfarawayfromlocaloptimal}), $z^*\neq y$.\footnote{For
proving the theorem, we
may assume $q>\sqrt{n}$ without loss of generality.
So $\Omega(\log q)$ is nonzero.}
So $z^*\in[2q+1]\setminus\{y\}$.
Now,
{\small % a little bit long 20211226 1:10
\begin{eqnarray*}
\sum_{x\in [n]}\,
d(z^*,x)
\ge
\sum_{x\in [n]\setminus[2q+1]}\,
d(z^*,x)
\stackrel{\text{(\ref{distancefunctionwithcopies})}}{=}
%\sum_{x\in [n]\setminus[2q+1]}\,
%d_{G^{(q)}}(z^*,x)\\
%=
\sum_{x\in [n]\setminus[2q+1]}\,
d_{G^{(q)}}(z^*,y)
\stackrel{\text{(\ref{localsolutionfarawayfromlocaloptimal})}}{=}
\Omega((n-(2q+1))\log q).
\end{eqnarray*}
}% a little bit long 20211226 1:10
%where the first equality follows from the points in $[n]\setminus[2q+1]$ being copies of $y$.
This and
equations~(\ref{boundinggoodtoallfirst})--(\ref{boundinggoodtoalllast})
%forbid
show
$z^*$
%(i.e., {\sf Alg}'s output)
to be no better than
$((\delta n/q)\cdot\log q)$-approximate for some constant $\delta>0$.
Clearly, $(\delta n/q)\cdot\log q=\omega(\log n)$.
%$o((n/q(n))\cdot\log(q(n)))$-approximate.
So taking $f(n)=\lfloor(n/q)\cdot(\log q)/(\log n)\rfloor$ completes the proof
except that $f(n)$ may be uncomputable.
%It remains to show the computability of $(\delta n/q)\cdot\log q$.
%Unfortunately, the worst-case query complexity $q$ may be achieved by any of the infinitely
%many metric spaces on $[n]$.
Gladly,
$d$ has codomain
$\{1/2^n,0,1,\ldots,n-1\}$ by equation~(\ref{distancefunctionwithcopies}).\footnote{Any graph on a subset of $[n]$ induces distances
in $\{0,1,\ldots,n-1,\infty\}$.
But equations~(\ref{boundinggoodtoallfirst})--(\ref{boundinggoodtoalllast}) forbid
$\infty$ as a distance.}
So we may pretend as if $q$ is {\sf Alg}'s worst-case query complexity
w.r.t.\ metrics with codomain $\{1/2^n,0,1,\ldots,n-1\}$.
%Therefore, we have already shown {\sf Alg} to be not better than
%$(f(n)\cdot\log n)$-approximate
%for some $f(n)=\omega(1)$,
%even when {\sc metric $1$-median}
%even
%when the codomain of the metrics
%with the codomain of the metrics restricted to $\{1/2^n,0,1,\ldots,n-1\}$.
This makes $q$, and thus $f(n)$, computable.
% the boring computability part 20220101 22:09
\comment{ % too complicated 20211226 19:21
By lines~1--3~and~9--12 of {\sf Adv} with $M$ (resp., $n$)
substituted by $[2q+1]$ (resp., $2q+1$), $G^{(q)}$
is a connected on $[2q+1]$.
So by equation~(\ref{distancefunctionwithcopies}),
$d$ has codomain
$\{0,1/2^n,1,2,\ldots,2q\}$.
Therefore, we have already shown {\sf Alg} to be not better than
$(f(n)\cdot\log n)$-approximate
for some $f(n)=\omega(1)$ even when {\sc metric $1$-median}
}% too complicated 20211226 19:21
\comment{ % incorporated the small distances really into the equation 20211226 19:09
The points in $[n]\setminus[2q+1]$ do not really need to be copies of $y$.
Just
set $d(y,a)$ and $d(a,b)$
to be an extremely small common value
for all distinct $a$, $b\in [n]\setminus[2q+1]$,
%let them be extremely close to $y$,
all other distances
remaining intact.
}% incorporated the small distances really into the equation 20211226 19:09
\end{proof}

%Theorems~\ref{maintheorem}~and~\ref{lowerboundextendedtoverysmallquerycomplexity}
%say that
%deterministic $o(n)$-query can

\begin{corollary}
{\sc Metric $1$-median} has no deterministic $o(n)$-query
$O(\log n)$-approximation algorithms.
\end{corollary}
\begin{proof}
Immediate from Theorem~\ref{lowerboundextendedtoverysmallquerycomplexity}.
\end{proof}

\comment{ % maybe it is clear from the statement 20211226 20:52
The following corollary
%refutes the existence of an asymptotically optimal
%says that
improves every
deterministic $o(n)$-query
algorithm
%s do not achieve
%approximation ratio.
asymptotically.
}% maybe it is clear from the statement 20211226 20:52

\begin{corollary}
{\sc Metric $1$-median} has no
deterministic $o(n)$-query
algorithms
%{\sc Metric $1$-median}
with an asymptotically best approximation ratio.
\end{corollary}
\begin{proof}
Take any deterministic $o(n)$-query
algorithm $A$.
By Theorem~\ref{lowerboundextendedtoverysmallquerycomplexity},
there exists a computable $f_A(n)=\omega(1)$
forbidding $A$ to be $o(f_A(n)\cdot\log n)$-approximate.
But Theorem~\ref{maintheorem}
asserts the existence of a
deterministic $o(n)$-query
$o(\sqrt{f_A(n)}\cdot\log n)$-approximation
algorithm.
\end{proof}

\appendix

\section{Distances in expanders}\label{expandersaveragedistanceleavingaset}

It is well-known that an $O(1)$-regular expander graph
$G^{\text{exp}}$
%=(M,E^{\text{Exp}})$
on $M$
exists.
I.e., there exist
constants $d\in\mathbb{Z}^+$ and $0<\alpha<1$
%$d=\Theta(1)$ and $\alpha=\Omega(1)$
such that
\begin{enumerate}[(i)]
\item $G^{\text{exp}}$ is $d$-regular, and
\item\label{expansionproperty}
for each
%nonempty
$S\subseteq M$ of size at most $n/2$,
at least $\alpha d\, |S|$ edges of $G^{\text{exp}}$
%cross $S$ (i.e., have exactly one
%endpoint in $S$ (and the other in $M\setminus S$).
are in $S\times (M\setminus S)$.
\end{enumerate}

%The following lemma, proved for completeness, is folklore.
%Below is an easy lemma.

\begin{lemma}\label{averagedistanceleavingaset}
For each nonempty $U\subseteq M$ of size at most $n/2$,
$$
%\frac{1}{|U|}
\sum_{x\in U}\, d_{G^{\text{\rm exp}}}\left(x,M\setminus U\right)=O(|U|).
$$
\end{lemma}
\begin{proof}
%Treat all vertices in $M\setminus U$ as a single vertex $\ell_0$.
%Then
%Starting from $\ell_0$,
%starting from which we
%run BFS
%Run BFS
%on $G^{\text{Exp}}$.
%, with the vertices in
%$M\setminus U$ treated as a single vertex, where BFS starts.
%with the vertices in $M\setminus U$ treated as a single vertex.
%Define $L_0\equiv M\setminus U$.
%So for
For
each $i\ge 1$,
%let
\begin{eqnarray*}
%L_0&\equiv& \{\ell_0\},\\
L_0&\equiv& M\setminus U,\\
L_i&\equiv& \left\{x\in U\mid d_{G^{\text{exp}}}\left(x,M\setminus U\right)=i\right\},\\
%is the set of vertices at level $i$ of the BFS tree.
%be the set of vertices in $U$ with distance $i$ from $M\setminus U$.
%Furthermore,
%Writing
S_i&\equiv&
%U\setminus(L_1\cup L_2\cup\cdots L_{i-1}).
L_i\cup L_{i+1}\cup\cdots
\end{eqnarray*}
So $L_i$
%(resp., $S_i$)
is the set of vertices at level $i$
%(resp., at least $i$)
of the BFS tree rooted at $M\setminus U$.\footnote{Generalize BFS in the obvious
way to allow the root to be a set of vertices.}
%(as an abuse of notation, let $L_0$ be $M\setminus U$ instead of $\{\ell_0\}$).

Now fix any $i\ge 1$.
%$S_i\times (M\setminus S_i)$
%$$S_i\equiv M\setminus(L_0\cup L_1\cup\cdots L_{i-1}).$$
%$$S_i\equiv L_0\cup L_1\cup\cdots L_{i-1}.$$
%Clearly, $d_{G^{\text{Exp}}}(a,M\setminus U)\le i-1$ for all $a\in S_i$.
%For each edge $(a,b)\in S_i\times (M\setminus S_i)$, where $a\in S_i$ and
%$b\notin S_i$,
%we have
%$b\in L_0\cup L_1\cup\cdot L_{i-1}$,
%implying
%$$d_{G^{\text{Exp}}}\left(b,M\setminus U\right)
%\le d_{G^{\text{Exp}}}\left(a,M\setminus U\right)+1\le i.$$
%By construction, $L_i$ is precisely the set of vertices in $S_i$
%incident to an edge in $(S_i\times (M\setminus S_i))$.
%Define $E_i\equiv (S_i\times (M\setminus S_i))\cap E^{\text{Exp}}$
%to be the set of edges of $G^{\text{Exp}}$ with exactly one endpoint in $S_i$.
%Elementary
%properties of BFS trees imply
%that all edges in $E_i$ have exactly one endpoint in $L_i$.
Because edges cannot cross
non-adjacent levels of a BFS tree,
%Elementary properties of BFS trees assert that
$S_i\times (M\setminus S_i)\subseteq L_i\times L_{i-1}$.
%\footnote{Edges cannot
%%only
%cross
%non-adjacent levels of a BFS tree.}
By item~(\ref{expansionproperty}) (with $S$ replaced by $S_i$ and noting
that $S_i\subseteq U$ has size at most $n/2$),
at least
%$|E_i|\ge \alpha d\, |S_i|$.
$\alpha d\, |S_i|$
edges of $G^{\text{exp}}$
are in $S_i\times (M\setminus S_i)$.
%have exactly one endpoint in $S_i$.
%I.e., at least $\alpha d\, |S_i|$ edges of $G^{\text{Exp}}$
%This and elementary
%In summary, at least $\alpha d\, |S_i|$ edges
%have
%exactly
%an endpoint in $L_i$.
%properties of BFS trees imply that
In summary,
at least $\alpha d\, |S_i|$ edges are in
%$L_i\times (M\setminus S_i)$.
$L_i\times L_{i-1}$
(and are thus incident to a vertex in $L_i$).
As $G^{\text{exp}}$ is $d$-regular, therefore, $|L_i|\ge \alpha\,|S_i|$.
%Equivalently,
Hence
\begin{eqnarray}
|S_{i+1}|=|S_i\setminus L_i|\le (1-\alpha)|S_i|.
\label{remainingvertices}
\end{eqnarray}

Iterating inequality~(\ref{remainingvertices}),
$$|S_j|\le (1-\alpha)^{j-1}|S_1|=(1-\alpha)^{j-1}|U|$$
for all $j\ge 1$.
%In particular,
So
\begin{eqnarray}
|L_j|\le|S_j|\le (1-\alpha)^{j-1}|U|\label{levelnottoolarge}
\end{eqnarray}
for all $j\ge 1$.
Now,
\begin{eqnarray*}
\sum_{x\in U}\, d_{G^{\text{exp}}}\left(x,M\setminus U\right)
&=&\sum_{j=1}^\infty\, \sum_{x\in L_j}\, d_{G^{\text{exp}}}\left(x,M\setminus U\right)\\
&=&\sum_{j=1}^\infty\, \sum_{x\in L_j}\, j\\
&=&\sum_{j=1}^\infty\, |L_j|\cdot j\\
&\stackrel{\text{(\ref{levelnottoolarge})}}{\le}&
\sum_{j=1}^\infty\, (1-\alpha)^{j-1}|U|\cdot j\\
&=&O(|U|),
\end{eqnarray*}
where the last equality uses the convergence of
$\sum_{j=1}^\infty\, (1-\alpha)^{j-1} j$.
%By construction, each edge in $S_i\times (M\setminus %S_i)$ connects a vertex in $S_i$
%with one in $(M\setminus U)\cup L_1\cup L_2\cup \cdots\cup L_{i-1}$.
%have
%$L_1\times (M\setminus S)$.
%exactly one endpoint in $S_i$.
%As $G^{\text{Exp}}$ is $d$-regular, therefore,
%$|L_1|\ge \alpha\,|S|$.
\end{proof}

\section{Acknowledgments}

The author is supported
by the Ministry of Science and Technology of Taiwan under
grant 110-2221-E-155-012-.

\bibliographystyle{plain}
\bibliography{median_lower_bound_few_queries}

\noindent

\end{document}